\documentclass[12pt,reqno,a4paper]{amsart}
\usepackage{amsmath,amssymb,amsfonts,amscd}
\usepackage[mathscr]{eucal}
\usepackage{hyperref,xcolor} 

\def\co{\colon\thinspace}

\DeclareMathOperator{\Dens}{\mathfrak{Dens}}

\newtheorem{theorem}{Theorem}

\newtheorem{proposition}{Proposition}

\theoremstyle{definition}
\newtheorem{example}{Example}
\newtheorem{remark}{Remark}

\newcommand{\R}{\mathbb{R}}

\newcommand{\la}{{\lambda}}

\newcommand{\ff}{{\varphi}}
\newcommand{\p}{\partial}
\newcommand{\w}{\hat{w}}

\newcommand{\lam}{\hat{\lambda}}

\newcommand{\g}{\gamma}

\newcommand{\wa}{{\widehat{\alpha}}}
\newcommand{\wb}{{\widehat{\beta}}}

\newcommand{\der}[2]{{\frac{\partial {#1}}{\partial {#2}}}}

\DeclareMathOperator{\Ber}{Ber}

\newcommand{\oder}[2]{{\frac{d {#1}}{d {#2}}}}
\newcommand{\RR}{\mathbb R}
\newcommand{\fun}{C^{\infty}}

\renewcommand{\t}{{\theta}}

\newcommand{\schw}{\mathfrak{S}}

\title[Factorization of generalized Sturm--Liouville operator]{Differential operators on the algebra of densities and factorization of  the generalized Sturm--Liouville operator}
\author{Ekaterina~Shemyakova}
\address{Department of Mathematics,  University of Toledo, Toledo,  Ohio, USA}
\email{ekaterina.shemyakova@utoledo.edu}
\author{Theodore~Voronov}
\address{School of Mathematics,  University of Manchester, Manchester,   UK}
\email{theodore.voronov@manchester.ac.uk}
\address{%
{Faculty of Physics, Tomsk State University, Tomsk, 634050, Russia}}

\date{15 (28) June    2018}

\begin{document}
\begin{abstract}
We consider factorization problem for differential operators on the commutative algebra of densities  (defined either algebraically or in terms of an auxiliary extended manifold)
introduced in 2004
by Khudaverdian and Voronov in connection with Batalin--Vilkovisky geometry.
We consider the case of the line, where unlike the familiar setting (where operators act on functions) there are obstructions for factorization. We analyze these obstructions.
In particular, we study the ``generalized Sturm--Liouville'' operators acting on the algebra of densities on the line.   This   in a certain sense is in between the 1D and 2D cases.
We establish a  criterion of factorizabily for the generalized Sturm--Liouville operator in terms of solution of the classical Sturm--Liouville equation. We also establish the possibility of an incomplete factorization.
\end{abstract}

\maketitle

\section{Introduction}

In this {letter} we consider differential operators on the algebra of densities. The algebra of densities is a {commutative algebra of functions of special form
defined on a formally extended (super)manifold, where a formal differentiation relative to the extra variable provides a new convenient way} to consider differential operators acting on densities of particular weights together. But also note that the original motivation for the development of the algebra of densities was the geometry of Batalin--Vilkovisky operators,
see the original paper~\cite{tv:laplace2}.

It has turned out that one may classify its derivations (obtaining an analog of Nijenhuis's theorem about derivations of differential forms), describe Poisson brackets, etc. It is also related with classical constructions of differential geometry such as projective connections and the notion of  {\textit{Thomas's bundle}}. See~\cite{tv:hovgroupoids, tv:algebradens}, \cite{jgeorge1},  \cite{biggs:lifting},  and \cite{hov:mattielap}.

{The interest in differential operators acting on densities comes from physics and also from integrable systems, and the work towards understanding   general structure of such operators has started in 90s.
Some classification results were obtained for the spaces of
differential operators of certain types acting on densities of different weights and considered as modules over Lie algebras of vector fields, and also for the mappings that relate these spaces,
e.g.~\cite{ovsienko_duval97,ovsienko_2007_circle}.}

{One of the works we pay more attention to in this note is
Duval and Ovsienko~\cite{ovsienko_duval97}, where
as part of such a classification of operators of order two and of weight zero, they discovered ``singularities" at weights $\la = 0, 1/2, 1$. An explanation of that was obtained in~\cite{biggs_hovik2013} in terms of \textit{canonical pencils} associated with self-adjoint operators on the algebra of densities introduced in~\cite{tv:laplace2} and motivated by BV geometry.
Duval and Ovsienko~\cite{ovsienko_duval97} also contains a
formula for the (unique) equivariant mapping connecting those spaces of differential operators (as modules over vector fields),  which led us into a new
natural derivation of the algebra of densities and a connection with this series of works, see  in~Sec.~\ref{subsec:duval}.
}

Factorization is a fundamental problem for ordinary and partial differential operators.
For operators in 1D, it is known that a factorization is always possible (in a large enough differential field extension)
and is closely related with solution of the corresponding differential equations. Starting from 2D, factorization is
{rarely possible}.
{Obstructions to factorizations describe algebraic properties} of the operators and are related with other important questions.
For example, the authors' particular further interest is the existence and types of Darboux transformations (see e.g. \cite{shemyakova2013_DT_fact, shem:darboux2}, \cite{shemya:hobby2016:iterated}, \cite{2013:invertible:darboux}). The case of  differential operators on the algebra of densities on the line is  in a certain sense ``in between'' the 1D and 2D cases, so it is interesting to see how an obstruction to factorization may arise. We show that unlike the 1D case, the operators on the algebra of densities $\Dens(\RR)$ are not always factorizable. Our main result is for \emph{generalized Sturm--Liouville operator},
which is a self-adjoint second order differential operator on the algebra $\Dens(\RR)$ on the line that specializes to the classical Sturm--Liouville when restricted on densities of weight $-1/2$.  We establish  a criterion for factorization of the  generalized Sturm--Liouville operator. It is formulated in terms of a density invariant (that we found) of the generalized Sturm--Liouville operator, and the criterion is that this density invariant should be  a solution of the corresponding classical Sturm--Liouville equation. We also establish the possibility of an incomplete factorization, which is unique.

{This letter} reports the first step in the project of unifying the study of algebraic properties of differential operators including factorization and Darboux transformations with their geometric analysis. Although the idea of the algebra of densities and differential operators on it first appeared in the context of supergeometry (from the study of odd Laplace operators), here we do not
consider the super case, which we
hope to address in the future (and to compare it with~\cite{shemya:voronov2016:berezinians}).

\section{Differential operators on the algebra of densities}

{The algebra of densities was discovered
by Khudaverdian and Voronov in~\cite{tv:laplace2} while
studying geometric structures related with the Batalin--Vilkovisky formalism. Here we recall
basic facts about this algebra, including differential operators  on it and their interpretation in terms of  the extended manifold, see~\cite{tv:laplace2,tv:algebradens}.} The generalized Sturm--Liouville operator on the algebra of densities  will be introduced in Sec.~\ref{subsec:genstuli}.

\subsection{Densities and the algebra of densities} Scalar and tensor densities are well known in differential geometry and  applications  in physics.
We shall deal with scalar densities. A (scalar) density of weight $\la\in \RR$
on a (super)manifold $M$ with local coordinates $x^a$ is a formal expression of the form
$\rho =   \rho(x) |Dx|^{\la}$.
The formal symbol $|Dx|$ under a change of coordinates transforms by
\begin{equation*}
    |Dx| = |Dx'| \, \left|\frac{Dx}{Dx'}\right|
\end{equation*}
and its formal power $|Dx|^{\la}$ transforms by
\begin{equation*}
    |Dx|^{\la}= |Dx'|^{\la}\, \left|\frac{Dx}{Dx'}\right|^{\la}\,.
\end{equation*}
Here we use the classical notation
${Dx}/{Dx'}$ for
the Jacobian, so  
\begin{equation*}
    \frac{Dx}{Dx'} =\det\der{x}{x'}
\end{equation*}
in the ordinary case
and
\begin{equation*}
    \frac{Dx}{Dx'} =\Ber\der{x}{x'}
\end{equation*}
in the supercase. (The main results of  this paper  are not concerned with supermanifolds, but we mention the supercase here for completeness;  also, the   algebra of densities, see below, arose first in the super context for the purposes of Batalin--Vilkovisky formalism.)
The line bundle whose sections are $\la$-densities is trivial topologically but is not trivialized canonically unless a special structure such as a volume element is fixed. A density of weight $0$ is a function. A density of weight $1$  is a (non-oriented) volume form, so  the integral $\int_M \rho(x) |Dx|$  over the manifold is well-defined, i.e.,  independent of a choice of coordinates; no orientation is required. ({One needs compact supports of course.}) For integral weights $\la=0,\pm 1,\pm 2, \ldots$, densities are sections of   tensor powers of the line bundle of volume forms. Densities of other weights   arise naturally in applications. Wave functions $\psi$ in quantum mechanics are densities of weight $\frac{1}{2}$ because   the squares of their  absolute values $|\psi|^2$ have the meaning of  probability densities and so  are densities of weight $1$. In an example that we shall consider, we shall naturally encounter densities of weights such as $-\frac{1}{2}$ and $\frac{3}{2}$.
Notation: $\Dens_{\la}(M)$ for the space of densities of weight $\la$ on $M$.

{Densities can be multiplied, so that the weights add
up:}
if $\psi=\psi(x)|Dx|^{\la}$ and $\chi=\chi(x)|Dx|^{\mu}$, then $\psi\chi=\psi(x)\chi(x)|Dx|^{\la+\mu}$.
In~\cite{tv:laplace2}, it was
observed that considering the direct sum
\begin{equation*}
    \Dens(M):= \bigoplus_{\la\in\RR} \Dens_{\la}(M)
\end{equation*}
as an algebra under such a multiplication extended by linearity to formal sums of densities of different weights  gives an  advantage similar to that  obtained by considering the  algebra of differential forms.
 The \emph{algebra of densities}  $\Dens(M)$ so defined can be analyzed in an abstract form as a commutative algebra with unit (e.g. considering its derivations and differential operators in the algebraic sense), as well as a particular algebra of functions (see more in the next subsection) where geometric  intuition can be applied. The algebra $\Dens(M)$ includes the algebra of functions,   $\fun(M)=\Dens_0(M)$, as a subalgebra, hence it in particular contains constants and the unit $1$.

The integral $I=\int_M$ can be extended by zero to densities of all weights $\neq 1$, giving a linear functional $I\co \Dens(M)\to \RR$,
which allows to
{extend the non-degenerate pairing of the spaces $\Dens_{\la}(M)$ and $\Dens_{1-\la}(M)$
(one can consider $\Dens_{1-\la}(M)$ as the ``smooth subspace''of the dual $\left(\Dens_{\la}(M)\right)^*$) to an invariant \emph{scalar product} on $\Dens(M)$\,,
\begin{equation}\label{eq.scalprod}
    \langle \psi, \chi\rangle := I(\psi\chi)\, ,
\end{equation}
where, strictly speaking, we consider the subalgebra of densities with compact support.}
More specifically, for $\psi\in\Dens_{\la}(M)$ and $\chi\in\Dens_{\mu}(M)$, define
\begin{equation}\label{eq.scalprod2}
    \langle \psi, \chi\rangle :=
    \begin{cases}
\int_M \psi\chi\,, & \la +\mu=1 \\
0 & \text{otherwise}
\end{cases}
\end{equation}
and then extend by bilinearity.

\subsection{{Densities as functions
on the extended manifold.}
}
\label{subsec:ovsienko}

The algebra $\Dens(M)$ can be interpreted~\cite{tv:laplace2} as a subalgebra of the algebra
of functions on an extended manifold $\hat M$, the total space of the frame bundle for the line bundle $|\Ber TM|$.
{The symbol $|Dx|$ can be taken as a basis element in the line bundle $|\Ber T^*M|$, which is dual to $|\Ber TM|$. It can be replaced by an invertible auxiliary variable $t$. }
The variables $x^a,t$ are local coordinates on $\hat M$ with the
transformation law $x^a=x^a(x'), t= |J| t'$, where $J=\Ber \der{x}{x'}$. Therefore an element of $\Dens(M)$ corresponds to a function of $x^a,t$,
\begin{equation*}
  \psi=  \sum \psi_{\la}(x)|Dx|^{\la} \longleftrightarrow \psi(x,t)= \sum \psi_{\la}(x)t^{\la} \ ,
\end{equation*}
a ``pseudo-polynomial'' in $t$.
By construction, the subalgebra $\Dens(M)\subset \fun(\hat M)$ is graded (by real numbers), unlike the whole algebra  $\fun(\hat M)$.

\begin{remark}  The fiber bundle $\hat M\to M$ was used in classical differential geometry  by T.~Y.~Thomas~\cite{thomas:projtheory} for  description of projective connections on a manifold $M$. The auxiliary manifold $\hat M$ is therefore sometimes referred to as the \emph{Thomas manifold} or \emph{Thomas bundle}. See  also~\cite{jgeorge1}. Unlike in that classical application, here the Thomas manifold $\hat M$ is taken  together with the graded algebra $\Dens(M)$ as its  algebra of functions and hence is regarded as a ``graded manifold''.
\end{remark}

Linear operators on the algebra $\Dens(M)$ and other objects associated with it naturally assume grading. They can be regarded as geometric objects on the graded manifold $\hat M$. Such are, for example, the graded derivations of $\Dens(M)$ (equivalently, graded vector fields on $\hat{M}$), which are classified  in~\cite{tv:laplace2} (see also \cite{tv:algebradens}).

\subsection{{Differential operators on the algebra of densities}}

Differential operators on $\Dens(M)$
can be introduced algebraically or in terms of local coordinates $x^a,t$ on $\hat M$. Consider the  operator $\w$
defined by
$\w (\rho) = \la \rho$ for $\rho \in \Dens_{\la}(M)$.
It is a derivation of   $\Dens(M)$ and so a differential operator of order $1$. It   is called the \emph{weight operator}. (Note a similarity with the  number   operator in quantum field theory.)  In coordinates,
\begin{equation}\label{eq.weightop}
    \w=t\der{}{t}\,.
\end{equation}
A linear operator $L$ on $\Dens(M)$ has weight $\mu$, i.e., for all $\la$, $L$ maps $\Dens_{\la}(M)$ to $\Dens_{\la+\mu}(M)$, if
\begin{equation}\label{eq.opweight}
    [\w, L]=\mu L\,.
\end{equation}
In particular, $\w$ itself has weight $0$ and because $t$ is invertible, the partial derivative in $t$ can be expressed via $\w$ as
\begin{equation*}
    \der{}{t}=t^{-1}\w\,.
\end{equation*}
Therefore
it is convenient to write {homogeneous}  differential operators on the algebra $\Dens (M)$ 
as
\begin{equation}\label{eq.do}
    L= t^{\mu}  \sum_{p+|q|\leq n} a_{pq}(x)\w^p\p_x^q\,,
\end{equation}
where $\p_x=\der{}{x^a}$ and $q$ is a multi-index. Formula~\eqref{eq.do} gives the (local) general form of a differential  operator of order $\leq n$ and weight $\mu$ on the algebra $\Dens (M)$\,. A special case of~\eqref{eq.opweight} is the commutation relation
\begin{equation}\label{eq.comm}
    [\w, t^{\mu}]=\mu t^{\mu}\,
\end{equation}
for the operator of multiplication by $t^{\mu}$. Note also the formulas for transformation under a change of coordinates $x^a=x^a(x')$, $t=t' |J|$: the operator $\w$ is invariant,  and the partial derivatives transform as
\begin{equation}\label{eq.derchange}
    \der{}{x^a}= \der{x^{a'}}{x^a}\, \der{}{x^{a'}} - \der{\ln |J|}{x^a}\,\w\,.
\end{equation}
Here $J=\Ber \der{x}{x'}$.

When a homogeneous operator on $\Dens(M)$ is restricted
on direct summands $\Dens_w(M)$, for all values of
$w\in \RR$, it gives a family of differential operators $\Dens_w(M)\to \Dens_{w+\mu}(M)$, where $\mu$ is the weight of $L$. Following~\cite{tv:laplace2}, we call
it
the \emph{operator pencil associated with $L$},
\begin{equation}\label{eq.dp}
    L_w= t^{\mu}  \sum_{p+|q|\leq n} a_{pq}(x)w^p\p_x^q\,.
\end{equation}
The \emph{order} of a pencil is counted jointly by the powers of $w$ and the order of the derivatives.

\begin{example}
An operator $L$ of order $1$ on $\Dens(M)$ corresponds to
 a linear pencil $L_w$:

\begin{equation}\label{eq.first}
    L= t^{\mu}\Bigl(A^a(x)\p_a+ B(x)\w +C(x)\Bigr)
    \longleftrightarrow
    L_w= |Dx|^{\mu}\Bigl(A^a(x)\p_a+ B(x)w +C(x)\Bigr) \,.
\end{equation}
Here $\p_a=\der{}{x^a}$.
\end{example}

When convenient, we may drop an indication to the manifold and write $\Dens$ or $\Dens_{\la}$.

\subsection{{Derivation of the algebra of densities from a different perspective.}}
\label{subsec:duval}

{The algebra of densities~\cite{tv:laplace2} was found in 2004 as the ``correct'' framework for the  Batalin--Vilkovisky $\Delta$-operator that uncovers
the geometry of this operator. The BV $\Delta$-operator on an odd symplectic manifold had generally been thought to be acting on functions and its definition required a choice of a  volume element~\cite{khudaverdian1991geometry}, \cite{schwarz1993geometry}, until it was discovered in 1999 that the  $\Delta$-operator is canonically defined  on semidensities~\cite{khudaverdian1999delta, khudaverdian2004semidensities}  without imposing any additional structure. Further investigations led from semidensities to the algebra of densities in~\cite{tv:laplace2}, where    the long-standing problem of describing all operators generating a given ``bracket" was fully solved.} H\"ormander's sub-principal symbol naturally fit  into this framework as a ``connection-type" geometric object (for densities of a fixed weight) and a part of the principal symbol on the extended manifold $\hat M$.

{In this letter, we are interested in other problems, but we have  adopted the algebra of densities as a convenient tool for working with differential operators acting on densities. Below we connect this approach with a large series of works by several authors who
studied interrelation between modules of differential operators acting on densities of different weights. A notable paper in this series is Duval and Ovsienko~\cite{ovsienko_duval97}, which considers the space $\mathcal{D}^2_\lambda$ of second-order linear differential operators of weight zero acting on}
densities of some particular weight $\lambda$,
\begin{equation} \label{eq:ovsienko_op_L}
L=a_2^{ij} \partial_i \partial_j + a_1^i \partial_i+a_0 \ , \quad
L: \Dens_\la \rightarrow \Dens_\la \ ,
\end{equation}
as a module over the Lie algebra of vector fields.
It turned out that for most  values of $\la$, all the    modules $\mathcal{D}^2_\lambda$ are isomorphic to each other. The only ``singular'' values of $\la$
are $\la=0,1/2,1$ for $\dim M \geq 2$ and
$\la=0,1$ for $\dim M = 1$. It was proved that there is a unique (up to a constant) equivariant linear mapping
\begin{equation*}
    \mathcal{L}_{\la\mu}^2:
    \mathcal{D}^2_{\mu} \rightarrow \mathcal{D}^2_{\la}
    \end{equation*}
and {coordinate} formulas for it were found. In particular, for $\dim \geq 2$,
the transformation $\mathcal{L}_{\la\mu}^2$ for
non-singular weights $\la$ and $\mu$  is given by
\begin{align*}
 &\widetilde{a}_2^{ij} = a_2^{ij} \ ,  \\
 &\widetilde{a}_1^l = \frac{2\la+1}{2\mu+1} a_1^l + 2 \frac{\mu-\la}{2\mu+1} \partial_i a_2^{il} \ , \\
& \widetilde{a}_0= \frac{\la(\la+1)}{\mu(\mu+1)} a_0 + \frac{\la(\mu-\la)}{(2\mu+1)(\mu+1)}
\left( \partial_i a_1^i - \partial_i \partial_j a_2^{ij} \right)  
\end{align*}
(formulas    do not depend on a choice of a local coordinate system).
These Duval--Ovsienko results can be used for the following  argument.
Since all the modules  $\mathcal{D}^2_\lambda$ corresponding to non-singular values of $\la$ are isomorphic, one is naturally led to the idea  that their elements should be studied as single objects depending on $\la$ as a parameter. So, if we fix some non-singular value of $\mu$
and some operator $L\in \mathcal{D}^2_\mu$, then by applying to it the Duval--Ovsienko transformation $\mathcal L_{\la \mu}$, we arrive at an  operator pencil $L_{\la}$ (i.e., a one-parameter family of operators $L_{\la}\in \mathcal D^2_{\la}$). 
 The general structure of the formula for $L_{\la}$ is that the
coefficients at $\partial_i \partial_j$
do not contain $\la$, the coefficients
at $\partial_i$ are linear in $\la$ and the free term  coefficients are quadratic in $\la$.
Therefore, our ``true objects'' obtained this way are characterized by the total order  in the derivatives $\partial_i$ and the parameter $\la$ together  (which should be $\leq 2$).

{Now, if we return to the algebra of densities, we see that to every differential operator on the algebra $\Dens$  corresponds
an operator pencil (see~\eqref{eq.dp}) which  has the desired structure of the coefficients automatically. Thus, if we write a second order operator on $\Dens$ as
\begin{equation*}
    \widehat{L} = A^{ij}_2 \p_i \p_j + A^i_1 \p_i + A_0  : \ \Dens \rightarrow \Dens \ ,
\end{equation*}
where $A^{ij}_2$ does not contain $\widehat{w}$, $A^i_1$ is of order $\leq 1$ in $\widehat{w}$,  $A_0$ is of order $\leq 2$ in $\widehat{w}$, and restrict the operator $\widehat{L}$ onto $\Dens_\la$,  we shall obtain differential operators $\Dens_{\la}\to \Dens_{\la}$, for each $\la$, whose coefficients depend on local coordinates $x^i$ and are polynomial in $\la$ with the same property as above.} 

Furthermore: the above   procedure based on Duval--Ovsienko formulas  gives operator  pencils that not only have  the  dependence on the parameter corresponding  to second-order differential operators on the algebra  $\Dens$, but a closer look at the formulas shows that these second-order operators on $\Dens$  also happen to be self-adjoint (i.e., exactly those studied in~\cite{tv:laplace2}). This was observed in~\cite{biggs_hovik2013}. The interesting question of what is happening at a singular value of parameter is considered in~\cite{hov:mattielap}.

\section{Factorization of differential operators on the algebra $\Dens(\R)$}

\subsection{Factorization of differential operators in the classical situation}

{Here we briefly overview the  factorization problem for differential operators in the classical setup. In it,   it is assumed that the operators act on functions, but  behaviour under changes of independent variables is not emphasized, and  the geometric nature of objects on which the operators act  is not really important in this context. We shall recall that the one-dimensional case and the higher dimensional case are very different.}

{
To illustrate  the situation in the one-dimensional case, consider the  simplest possible example of the  operator $L=\p^2$. It has a commutative factorization $L=\p \circ \p$ and
a one-parameter family of non-commutative factorizations
\begin{equation*}
    L=\p^2 = \left( \p + \frac{1}{x+c} \right) \circ \left( \p - \frac{1}{x+c} \right) \ .
\end{equation*}
One can show that these are the only possible factorizations.}
{In general,  
for an operator  $\displaystyle L= a_n \p^n + a_{n-1} \p^{n-1} + \dots  + a_1 \p + a_0$, where $a_i$ do not need to be constants,  there is Frobenius's theorem~\cite{frobenius} stating that factorizations into first-order factors are in one-to-one correspondence with maximal flags in the kernel of $L$.
For example, continuing with the operator  $L=\p^2$, one gets for a particular flag
\begin{equation*}
 \langle x\rangle \subseteq \langle  x, 1 \rangle = \ker L \quad \longleftrightarrow \quad   L= \left( \p + \frac{1}{x} \right) \circ \left( \p - \frac{1}{x} \right)\,,
\end{equation*}
which can be computed as $L=L_2 \circ L_1$, where $L_1 = x \circ \p \circ (x)^{-1}$ and then $\varphi_2=L_1(1)=-1/x$, and
$L_2=\varphi_2 \circ \p \circ \varphi_2^{-1}$.
}

{
Already in two dimensions, the situation is radically  different, as shown by the
famous  example of E.~Landau: for $P=\p_x+x\p_y$, $Q=\p_x+1$, $R=\p_x^2+x\p_x \p_y + \p_x +(2+x)\p_y$ (irreducible in any reasonable extension), we have
$QQP=RQ$, which means that an operator of third order has two completely different factorizations into irreducible factors. Noteworthy, even the numbers of the irreducible factors are different. In general, in the multidimensional case there is no general theory, but only some results for  particular situations that are available.
}

{Factorization of differential operators is a fundamental problem in mathematics, and there are, very broadly speaking, three main aspects here: analytic (where the  objects on which the operators act are analyzed as functions satisfying this or that analytic properties), algebraic and computational (where the  objects are treated as elements of some differential field and so all the  questions about smoothness, etc., are skipped), and geometric (where the geometric nature of objects is important, e.g. scalar  functions, densities,  differential forms, etc.).}

{Our   interest here is geometric and we shall be concerned with the factorization problem in the framework of the algebra of densities. In~\cite{tv:darbouxsuperline,shemya:voronov2016:berezinians},
among other results we considered factorization of differential operators on the superline (i.e., a manifold of dimension $1|1$). It turned out that for \textit{non-degenerate operators} of arbitrary order an analog  of the Frobenius theorem holds. Further,  in~\cite{tv:darbouxsuperline,shemya:voronov2016:berezinians}, it was shown that non-degenerate differential  operators in  dimension $1|1$ are similar to operators  in the one-dimensional case.}

\subsection{Factorization of differential operators on the algebra of densities. Obstruction}

{If we restrict ourselves to the algebra $\Dens(\R)$  
and   operators of order two on it (which we think is not such a restriction as the ideas for operators of orders larger than two probably will be similar to the ones of order two),
we shall  discover that the  situation with factorization of differential operators is  
very different here.}

{Now, the most general form for  factorization here will
be as follows:}
\begin{align}
    L & = t^{2}  \left(\partial^2 + \left(p_1 \hat{w} + p_0 \right) \partial +  q_2 \hat{w}^2 + q_1 \hat{w} + q_0 \right) \label{eq:L} \\
      & = t
\Bigl(  \partial - \alpha_1 \hat{w} - \alpha_0  \Bigr) \cdot t
\Bigl(  \partial - \beta_1 \hat{w}- \beta_0 \Bigr) \\
      & = t^2 \Bigl(  \partial - \alpha_1 (\hat{w}+1) - \alpha_0  \Bigr) \cdot
\Bigl(  \partial - \beta_1 \hat{w}- \beta_0 \Bigr) \ ,
\label{eq:fact}
\end{align}
where $p_i, q_i, \alpha_i, \beta_i$ are functions of $x$ and where the commutation rule  $\hat{w} t^\la=t^\la (\hat{w}+\la)$ was used in the last line. The negative signs are for convenience of further computations. Note that we consider monic operators of matching weights to ensure invariance under changes of coordinates.

{Up to the factor of  $t^2$, this may be viewed as a problem of factoring a univariate differential operator acting on functions of $x$, with coefficients depending on a parameter $w$, of a special form dictated by the homogeneity condition. So, morally it is ``in between" 1D and 2D cases.}

{To explore the similarity with the 1D case,} we switch to a  condensed notation $\hat{\alpha}=- \alpha_1 (\hat{w}+1) - \alpha_0$, $\hat{\beta}=- \beta_1 \hat{w}- \beta_0$, $\hat{p}=p_1 \hat{w} + p_0$, $\hat{q}=q_1 \hat{w} + q_0$, which hides
dependence on $\hat{w}$ and makes the problem look similar to the problem in the classical setting:
\begin{align*}
    L & = t^{2}  \Bigl(\partial^2 + \hat{p} \partial +  \hat{q} \Bigr) = t^2 \Bigl(\partial - \widehat{\alpha}\Bigr) \cdot
\Bigl(\partial - \widehat{\beta}\Bigr)
\end{align*}
This implies
$\wa=-\hat{p}-\wb$ and the familiar Riccati equation for $\wb$:
\begin{equation}
  -\wb'  = \wb^2 + \hat{p} \wb + \hat{q} \ .
\end{equation}

Warning: here  we are  applying  prime for the derivative in $x$  (and not for transformation of variables!) for a more expedient notation.
In the classical case, substitution
$
\wb =   \partial \left( \ln \ff \right)
$
is used to transform the Riccati equation into
\begin{equation} \label{eq:Leq}
 \ff'' + \hat{p} \ff' + \hat{q} \ff = 0 \ .
\end{equation}
That leads to the known fact that in the {classical situation}, every $\ff$ in the kernel of $L$ implies a factorization of $L=\partial^2 + \hat{p} \partial +  \hat{q}$.  {So in a suitable field of coefficients all 1D differential operators are factorizable.}

In the {densities case}, the
operator is factorizable \textit{only if} there is a solution $\ff$ of~\eqref{eq:Leq} of a certain special form, namely,  such that $\wb$ computed as $\wb =   \partial \left( \ln \ff \right)$ is linear in $\hat{w}$.

\begin{example}  The following operator on   the algebra of densities (which is a generalized Sturm--Liouville operator, see next section)  is \emph{not} factorizable:
$$\displaystyle L = t^{2} \cdot \left(\partial^2 -   \hat{w}^2 - \hat{w}  \right)\,.$$
See more in Sec.~\ref{sec:genSL} .
\end{example}

Let us now return to the non-condensed notation, and look
more closely at this Riccati. We had
\begin{align*}
    L & = t^{2}  \left(\partial^2 + \left(p_1 \hat{w} + p_0 \right) \partial +  q_2 \hat{w}^2 + q_1 \hat{w} + q_0 \right) \\
&= t^2 \Bigl(  \partial - \alpha_1 (\hat{w}+1) - \alpha_0  \Bigr) \cdot
\Bigl(  \partial - \beta_1 \hat{w}- \beta_0 \Bigr)
\end{align*}
After equating the coefficients we have: $\alpha_1=-p_1 - \beta_1$, $\alpha_0=
-p_0 + p_1 + \beta_1- \beta_0$,
and the single Riccati equation for $\wb=\beta_0+\beta_1\hat{w}$ splits into the three equations
\begin{align*}
    0&= \beta_1^2+ p_1\beta_1+q_2 \\
    - \beta_1'&=  2\beta_0\beta_1 + p_0\beta_1+p_1\beta_0+q_1 \\
    -\beta_0'&=
    \beta_0^2+p_0\beta_0+q_0
\end{align*}
So we have an overdetermined system in $\beta_0$ and $\beta_1$, and there should be
some compatibility condition of  vanishing of the obstruction to factorization.

\section{Generalized Sturm--Liouville operator} \label{sec.GSLO}

\subsection{The Khudaverdian--Voronov {canonical operator pencil}.}

Return temporarily to arbitrary dimension.
We shall  use one important feature of $\Dens(M)$, which is the invariant scalar product. This gives the
possibility of considering (formal) adjoint and self-adjoint operators.
First of all,
\begin{equation*}
    \w^*=1-\w\,, \quad \left(\p_a\right)^*=- \p_a\,, \quad t^*=t\,, \quad (f(x))^*=f(x),
\end{equation*}
(where $\p_a=\der{}{x^a}$).

In the sequel, we write all the  formulas for ordinary manifolds and even operators on them, for simplicity. (Though we still use {$Dx$} {as the notation for the coordinate volume element.
It is possible to write the formulas for the super case too, see~\cite{tv:laplace2}.)}

For a first order operator of weight $\mu$ given by~\eqref{eq.first},
we have
\begin{equation}\label{eq.firststar}
    L^*= t^{\mu}\Bigl(-A^a(x)\p_a-B(x)\w  -\p_aA^a(x)+(1-\mu)B(x) +C(x)\Bigr)\,.
\end{equation}
Hence it cannot be self-adjoint unless $A^a\equiv 0$ and $B\equiv 0$ (when it is effectively a zero-order operator). It is anti-self-adjoint if $-C=C-\p_aA^a-(1-\mu)B$\,. Note that $C=L(1)$. If there are no a priori conditions for $C$, then this simply defines $C$ in terms of $A^a$ and $B$, which are not subject to any constraints,
\begin{equation*}
    C=\frac{1}{2}\bigl(\p_aA^a -(1-\mu)B\bigr)\,.
\end{equation*}
If however we consider the particular case of operators with $C\equiv 0$, i.e., derivations
 \begin{equation}\label{eq.der}
    L= t^{\mu}\Bigl(A^a(x)\p_a+ B(x)\w\Bigr)\,,
\end{equation}
then $L^*=-L$ if and only if
\begin{equation*}
   \p_aA^a = (1-\mu)B\,,
\end{equation*}
which has different meanings for $\mu=1$ and for general $\mu\neq 1$. In the first case, we have the invariant ``divergence-free'' condition $\p_aA^a=0$ for the vector density $A=|Dx|A^a\p_a$ on $M$ (of weight $1$) and no restrictions for $B$. In the second case, $B$ is defined by $A^a$, and we have operators of the form
\begin{equation}
    L= t^{\mu}\Bigl(A^a(x)\p_a+ \frac{\p_aA^a}{1-\mu}\,\w\Bigr)\,,
\end{equation}
which are interpreted as ``generalized Lie derivatives''. (This is part of classification of the derivations of the algebra $\Dens(M)$ obtained in~\cite{tv:laplace2, tv:algebradens}.)

{For   second-order self-adjoint operators on $\Dens(M)$, the  analysis was done in~\cite{tv:laplace2} (see also~\cite{tv:algebradens}).}

\begin{theorem}[\cite{tv:laplace2}]
A self-adjoint second-order operator $L$ on the algebra of densities   $\Dens(M)$ satisfying the condition $L(1)=0$ (annihilating constants) is completely defined by the following data:
\begin{enumerate}
  \item a tensor density $S=|Dx|^{\mu}S^{ab}\p_a\otimes \p_b$ (the principal symbol of the restriction of $L$  on  functions);
  \item an object $\g^a$ interpreted as   ``upper connection'' coefficients, associated with $S$;
  \item an object $\theta$ called ``Brans--Dicke field'' (a terminology from physics), associated with $S$ and $\gamma^a$.
\end{enumerate}
The operator defined by such data is given by
\begin{equation}\label{eq.canoper}
    L=t^{\mu}\Bigl(S^{ab}\p_a\p_b +\bigl(2\g^a\w  + \p_aS^{ab}+(\mu-1)\g^b\bigr)\p_b +\theta \w^2+\bigl(\p_a\g^a+(\mu-1)\t\bigr)\w\Bigr)\,.
\end{equation}
The symmetric matrix
\begin{equation}
\hat S^{\hat a\hat b}=t^{\mu}
    \begin{pmatrix}
      S^{ab} & \g^a \\
      \g^b & \theta \\
    \end{pmatrix}
\end{equation}
gives the principal symbol of $L$ (as a quadratic form) on the extended manifold $\hat M$.
\end{theorem}

(Grouping $S^{ab}$, $\g^a$, $\t$ into the above matrix is the explanation of why the operator $L$ is parameterized this way. The theorem says, in brief: a second-order self-adjoint operator on the algebra $\Dens(M)$ annihilating constants is completely defined by its principal symbol.)

The coefficients $S^{ab}$, $\g^a$, $\t$ as geometric objects on $M$ are characterized  by their transformation laws under a change of coordinates: in ``new'' coordinates indicated by dash,
\begin{align}\label{eq.translawsgt}
    S^{a'b'}&=S^{ab}\,\der{x^{a'}}{x^a}\,\der{x^{b'}}{x^b}\,|J|^{\mu}\\
    \g^{a'}&= \Bigl(\g^{a}-S^{ab}\p_{b}\ln|J|\Bigr)\,\der{x^{a'}}{x^a}\,|J|^{\mu}\\
    \t'&=\Bigl(\t - 2\g^{a}\p_{a}\ln|J|+ S^{ab}\,\p_{a}\ln|J|\cdot\p_{b}\ln|J|\Bigr)\,|J|^{\mu}\,.
\end{align}
Here $J=\Ber\der{x}{x'}$.

\begin{example}  Objects of type $\g^a$ and $\t$ can be constructed from a  connection in the bundle of volume forms on $M$, as follows:
 \begin{align*}
    \g^a :=S^{ab}\g_a\\
    \t :=S^{ab} \g_a\g_b\,.
 \end{align*}
(Here $\g_a$ are the connection coefficients.)
If   $S^{ab}$ is invertible, then conversely, a connection $\g_a$ can be obtained from $\g^a$. But there will still be a freedom in defining $\t$.
(For a full analysis of the degrees of freedom there, see~\cite{tv:laplace2, tv:hovgroupoids}  and~\cite{hov:mattielap}.)
\end{example}

A self-adjoint operator $L$ on the algebra $\Dens(M)$ defined by~\eqref{eq.canoper} and specified by data $S^{ab}, \g^a, \t$ is called \emph{canonical operator}. (Its canonicity is in the fact that it requires less data than a general second-order operator and these data have a geometric meaning.) The corresponding quadratic operator pencil
\begin{equation}\label{eq.canpencil}
    L_w=|Dx|^{\mu}\Bigl(S^{ab}\p_a\p_b +\bigl(2\g^aw  + \p_aS^{ab}+(\mu-1)\g^b\bigr)\p_b +\theta w^2+\bigl(\p_a\g^a+(\mu-1)\t\bigr)w\Bigr)
\end{equation}
($w$ is the parameter)
is referred to as a \emph{canonical pencil}.

\subsection{Generalized Sturm--Liouville operator.}
\label{subsec:genstuli}
It is classically well-known (see~\cite{hitchin:hsw}) that the familiar Sturm--Liouville operator  $\p^2+u(x)$ on the line can be  uniquely described  by the requirements  that it keeps its form under a change of coordinate and is self-adjoint. (In particular, keeping the form without the first derivative implies that it \emph{cannot} be regarded as acting on functions.) Together these conditions force it to be
of the following form:
\begin{equation}\label{eq.classstuli}
   L_{-\frac{1}{2}}= |dx|^2\Bigl(\p^2 +u(x)\Bigr) : \quad \Dens_{-\frac{1}{2}}\to \Dens_{\frac{3}{2}} \, .
\end{equation}
(We   use $dx$ instead of $Dx$ for the case of the line.)

Comparison with the above indicates that this classical operator should be seen as the restriction of a canonical pencil  on the line.
This was first observed  in~\cite[\S 4.1]{tv:laplace2}. We can elaborate that as follows. 
\begin{proposition}
The form of a second order operator $L$ of weight $2$ on the algebra $\Dens(\RR)$, 
\begin{equation}\label{eq.genstuli}
    L=|dx|^{2}\Bigl(\p^2 +\g\bigl(2\w   +1\bigr)\p  +  \bigl(\t(\w +1)+ \g_x\bigr)\w\Bigr)\,,
\end{equation}
is invariant under changes of the independent variable $x$.  It is uniquely defined by the conditions for $L$ being monic (i.e., the coefficient $1$ at $\p^2$), self-adjoint and annihilating constants. The operator $L$ is specified by   geometric quantities $\g$ and $\t$ 
and specializes to the classical Sturm--Liouville operator at $w=-1/2$. 
\end{proposition}
\begin{proof}
Indeed, if we specialize  formulas~\eqref{eq.canoper} and~\eqref{eq.canpencil} for 1D, we shall obtain
\begin{equation}\label{eq.canoperline}
    L=t^{\mu}\Bigl(S\,\p^2 +\bigl(2\g\w  + \p S +(\mu-1)\g\bigr)\p  +\t  \w^2+\bigl(\p\g+(\mu-1)\t\bigr)\w\Bigr)\,,
\end{equation}
and
\begin{equation}\label{eq.canpencilline}
    L_w=|dx|^{\mu}\Bigl(S\,\p^2 +\bigl(2\g w  + \p S +(\mu-1)\g\bigr)\p  +\t   w^2+\bigl(\p\g+(\mu-1)\t\bigr) w\Bigr)
\end{equation}
for the corresponding operator pencil. Here   $\p=\p_x$ (differentiation in the single variable),  $\p S=\p(S)$, and  $\p\g=\p(\g)$). The transformation of the coefficients $S, \g, \t$ simplifies and takes the form:
\begin{align}\label{eq.translawline}
    S' &=S\,J^{\mu-2}\\
     \g'&= \Bigl(\g -S\,\p\ln J \Bigr)\,J^{\mu-1}\\
    \t'&=\Bigl(\t - 2\g\, \p\ln J + S\,(\p\ln J)^2\Bigr)\,J^{\mu}\,.
\end{align}
We have assumed that $J>0$ and have therefore replaced $|J|$ by $J$, which in this 1D case coincides with the derivative $\oder{x}{x'}$. It will be clarifying to write  the derivatives such as $\p f$ as $f_x$ indicating the variable explicitly and to elaborate further $\p \ln J$.
Since $J=\oder{x}{x'}$, we rewrite it via $J^{-1}=\oder{x'}{x}=x'_x$. Finally, we obtain:
\begin{align}\label{eq.translawlinesimpl}
    S' &=S\,(x'_x)^{2-\mu}\\
     \g'&= \Bigl(\g\,x'_x +S\,x'_{xx} \Bigr)\,(x'_x)^{-\mu}\\
    \t'&=\Bigl(\t\,(x'_x)^2 + 2\g\, x'_x x'_{xx} + S\,(x'_{xx})^2\Bigr)\,(x'_x)^{-\mu-2}\,.
\end{align}
In particular, if the operator $L$ has weight $\mu=2$, the top coefficient $S$ becomes invariant. (A peculiarity of   1D.) We may choose $S$ to be $1$ and hence arrive at~\eqref{eq.genstuli}. Clearly, at $w=-1/2$ the term with the first derivative vanishes and we obtain an operator of the classical Sturm--Liouville form.
\end{proof}

We shall call such an operator $L$ on $\Dens(\RR)$ given by~\eqref{eq.genstuli}, the \emph{generalized Sturm--Liouville operator}. The coefficients $\g$ and $\t$ of the generalized Sturm--Liouville operator transform by
\begin{align}\label{eq.translawgenstuli1}
     \g'&= \Bigl(\g\,x'_x + \,x'_{xx} \Bigr)\,(x'_x)^{-2}\\
    \t'&=\Bigl(\t\,(x'_x)^2 + 2\g\, x'_x x'_{xx} +  (x'_{xx})^2\Bigr)\,(x'_x)^{-4}\,. \label{eq.translawgenstuli2}
\end{align}
The restriction of $L$ on $\Dens_{-\frac{1}{2}}$ gives the classical Sturm--Liouville operator 
\eqref{eq.classstuli} mapping  $\Dens_{-\frac{1}{2}}$ to  $\Dens_{\frac{3}{2}}$\,,
with the potential $u=u(x)$   expressed via $\g$, $\t$ by
\begin{equation}\label{eq.u}
    u= -\frac{1}{2}\Bigl(\g_x+\frac{1}{2}\,\t\Bigr)\,.
\end{equation}

\begin{proposition} 
The transformation laws for $\g$ and $\t$ induce  the well-known  transformation law for the potential $u$,
\begin{equation}\label{eq.translawu}
    u'=\Bigl(u - \frac{1}{2}\schw_x x'\Bigr)\,(x'_x)^{-2}\,,
\end{equation}
where $\schw_x f$ denotes the Schwarzian  derivative of a function $f$ with respect to a variable $x$.
\end{proposition}

\begin{proof} Consider the transformation laws for $\g$ and $\t$ given by~\eqref{eq.translawgenstuli1} and \eqref{eq.translawgenstuli2}. By differentiating   both sides of~\eqref{eq.translawgenstuli1} with respect to $x'$, we obtain
\begin{equation*}
    \g'_{x'}= \Bigl(\g_x(x'_x)^2 - \g x'_{x}x'_{xx} + x'_{xxx}x'_x - 2 (x'_{xx})^2\Bigr)(x'_{x})^{-4}\,.
\end{equation*}
Combining this with~\eqref{eq.translawgenstuli2}, we obtain directly
\begin{multline*}
    \g'_{x'}+\frac{1}{2}\,\t'= \Bigl(\g_x(x'_x)^2 - \g x'_{x}x'_{xx} + x'_{xxx}x'_x - 2 (x'_{xx})^2
    +   \frac{1}{2}\,\t (x'_x)^2 +  \g x'_{x}x'_{xx} +  \frac{1}{2}\,(x'_{xx})^2 \Bigr)(x'_{x})^{-4}=\\
    \Bigl(\bigl(\g_x+ \frac{1}{2}\,\t\bigr)(x'_x)^2   + x'_{xxx}x'_x - \frac{3}{2}\, (x'_{xx})^2\Bigr)(x'_{x})^{-4}=
    \Bigl(\g_x+ \frac{1}{2}\,\t    + \frac{x'_{xxx}}{x'_x} - \frac{3}{2}\, \Bigl(\frac{x'_{xx}}{x'_x}\Bigr)^{\!2}\;\Bigr)(x'_{x})^{-2}\,,
\end{multline*}
which is exactly
\begin{equation*}
    \g'_{x'}+\frac{1}{2}\,\t'= \Bigl(\g_x+ \frac{1}{2}\,\t    + \schw_xx'\;\Bigr)(x'_{x})^{-2}\,,
\end{equation*}
by the definition of the Schwarzian derivative:
\begin{equation*}
    \schw_x f = \frac{f_{xxx}}{f_x} - \frac{3}{2}\, \Bigl(\frac{f_{xx}}{f_x}\Bigr)^{\!2}\,,
\end{equation*}
for a function $f=f(x)$ (see, e.g.~\cite{ovsienko-tabachnikov:book2005}). This gives~\eqref{eq.translawu}\,.
\end{proof}

(The appearance of the Schwarzian derivative in the  transformation of the potential for the   Sturm--Liouville equation is  classical, see e.g.~\cite[\S 3.8]{hitchin:hsw}. Its geometric interpretation can be found in~\cite{ovsienko-tabachnikov:book2005}. The fact that this transformation law is a corollary for the 1D case of the transformation laws of an ``upper connection'' on volume forms $\g^{a}$ and  ``Brans--Dicke field'' $\t$ was discovered in~\cite{tv:laplace2}.)


\subsection{Factorization of the generalized Sturm-Liouville operator}
\label{sec:genSL}

In Sec.~\ref{subsec:duval} we observed that when modules of differential operators acting on densities of non-singular weights are studied, the  families of operators that  appear  correspond to self-adjoint operators   on the algebra of densities. There is one per family  self-adjoint operator acting between densities of a particular combination of weights ~\cite{tv:laplace2}. It corresponds to a singular value of the parameter.
Now, of course,   operators on the algebra of densities can be specialized to  operators acting on  densities of singular weights, just there will be no one-to-one correspondence.

Consider~\eqref{eq.genstuli}, {the generalized Sturm--Liouville operator}.
When specialized on
$\Dens_{-1/2}$, it becomes the classical Sturm-Liouville operator
$t^2\left( \p^2 +u \right)$
with the  potential
\begin{equation*}
   u=-\frac{1}{2}\left(\g_x+\frac{\theta}{2}\right) \,,
\end{equation*}
as we have just discussed. 
Rewriting everything in terms of $u$ instead of $\theta$ and an operator $\lam$ such  that $\w=-\frac{1}{2}+\lam$, we come to
\begin{equation} \label{eq.genstuli2}
  L=t^2\left(\p^2+ {2\g\lam}\p -2 (\g_x+2u)\lam^2 + \g_x\lam +  u \right)
\end{equation}
where the classical Sturm--Liouville now corresponds to $\la=0$.

\begin{theorem} \label{thm:main} The necessary and sufficient condition for the generalized Sturm-Liouville operator $L$ given by~\eqref{eq.genstuli} or~\eqref{eq.genstuli2} to be factorizable is that the quantity $\psi$ defined by
\begin{equation}\label{eq.psi}
    \psi:= \left(\g^2+2(\g'+2u)\right)^{-1/4}
\end{equation}
satisfies the classical Sturm-Liouville equation
\begin{equation}
    (\p^2+u)\psi=0\,.
\end{equation}
\end{theorem}
\begin{proof}
Consider the factorization problem $L=t^2(\p-\wa)(\p-\wb)$, which,
according to the general analysis in the previous section, is  parameterized by $\wb=b_0+b_1\lam$ and the remaining conditions are
\begin{align*}
    -b_0'&=b_0^2+u\\
    -b_1'&=2b_0 b_1 +2\g b_0+\g'\\
    0&=b_1^2+2\g b_1-2(\g'+2u)
\end{align*}
(here prime is used for the derivative).
From the third equation we get:
\begin{equation}\label{eq:4}
    b_1+\g=\pm \sqrt{\g^2+2(\g'+2u)}\,.
\end{equation}
$b_1+\g$ is exactly what we need for the second equation, which has the form
\begin{equation} \label{eq:5}
    -(b_1+\g)'=2b_0(b_1+\g)\, .
\end{equation}
But first let us solve~\eqref{eq:5} for $b_0$ (and then substitute the expression from \eqref{eq:4}):
\begin{multline*}
    b_0=-\frac{1}{2}\p\ln (b_1+\g)=
    -\frac{1}{2}\frac{\left(\pm\sqrt{\g^2+2(\g'+2u)}\right)'}{\pm\sqrt{\g^2+2(\g'+2u)}}=\\
    -\frac{1}{2}\p\ln\sqrt{\g^2+2(\g'+2u)}=
    \p\ln\left(\g^2+2(\g'+2u)\right)^{-1/4}\,.
\end{multline*}
Denoting
\begin{equation*} 
    \psi:= \left(\g^2+2(\g'+2u)\right)^{-1/4}
    \,,
\end{equation*}
the previous expression for $b_0$ becomes
$b_0=\p\ln\psi$.
What is left, is the first equation from the original system, which is the ``classical'' Riccati equation:
\begin{equation*}
    -b_0'=b_0^2+u\,,
\end{equation*}
equivalent to $\psi$ satisfying the   Sturm--Liouville equation with the potential $u$\,. 
\end{proof}

One immediate question concerns   invariance of the obtained condition under a change of coordinates. For that we need to consider the transformation law for the object $\psi$. Note that if we express back $u$ via $\g, \t$, then the expression for $\psi$ takes the form
\begin{equation}\label{eq.psiviagt}
    \psi= \left(\g^2-\t\right)^{-1/4}\,.
\end{equation}

\begin{theorem} The $-\frac{1}{2}$-density
\begin{equation*}
    \boldsymbol{\psi} := \left(\g^2-\t\right)^{-1/4} |dx|^{-1/2}
\end{equation*} 
does not depend on a  change of coordinate and hence
is an invariant of the generalized Sturm-Liouville operator $L$.
\end{theorem}
\begin{proof}
We need to recall the transformation laws~\eqref{eq.translawgenstuli1} and~\eqref{eq.translawgenstuli2}. We   return to the notation where prime is used for ``new'' variables and quantities in new variables. We shall again use subscripts for derivatives. From~\eqref{eq.translawgenstuli1} and~\eqref{eq.translawgenstuli2}, we have
\begin{align*}
     (\g')^2&= \Bigl(\g^2\,(x'_x)^2 + 2\g\,x'_x\,x'_{xx} +(x'_{xx})^2\Bigr)\,(x'_x)^{-4}\,,\\
    \t'&=\Bigl(\t\,(x'_x)^2 + 2\g\, x'_x x'_{xx} +  (x'_{xx})^2\Bigr)\,(x'_x)^{-4}\,,
\end{align*}
and by subtracting we obtain after simplification
\begin{equation*}
    {\g'}^2-\t'=(\g^2-\t)\,(x'_x)^{-2}\,.
\end{equation*}
Hence the quantity $\boldsymbol{\omega}=(\g^2-\t)\,|dx|^2$ is invariant. It is a well-defined density of weight $2$. Therefore the quantity
\begin{equation*}
   \boldsymbol{\psi}= \psi \,|dx|^{-1/2}=\boldsymbol{\omega}^{-1/4}
\end{equation*}
is a well-defined density of weight $-1/2$. 
\end{proof}

By the properties of the classical Sturm--Liouville operator, the condition that 
$\boldsymbol{\psi}$ is a solution does not depend on a choice of    coordinate. This establishes the invariance of the factorization criterion.

For the same operator $L$ on the algebra $\Dens(\R)$, we can pose the problem of an ``incomplete factorization'':
\begin{equation*}
    L=t^2\left((\p-\wa)(\p-\wb)+f\right)\,,
\end{equation*}
where $f=f(x)$ does not contain $\w$. 
This   gives the system
\begin{align*}    -b_0'&=b_0^2+u-f\\   -b_1'&=2b_0 b_1 +2\g b_0+\g'\\  0&=b_1^2+2\g b_1-2(\g'+2u)
\end{align*}
for the unknowns $b_0$, $b_1$ and $f$.
A slight modification of the argument above leads us to the following theorem.

\begin{theorem} An incomplete factorization of the generalized Sturm--Liouville operator is always possible and  is unique. It  is given by the formulas
\begin{align*}
    b_0&=\p\ln\psi\,,\\
    b_1&=-\g\pm\frac{1}{\psi^2}\,,\\
    f&=\frac{1}{\psi}(\p^2+u)\psi\,.
\end{align*}
Here $\psi$ is as above. \emph{(We assume that $\psi$ is not identically zero.)} \qed
\end{theorem}

One can draw some indirect analogy for these results with the Laplace (differential) gauge invariants $h$ and $k$, which
appear in the study of operators of the form
$\p_x \p_y + a\p_x+b\p_y +c$ in the classical setting,
the direction originally started in  Darboux's famous book~\cite{Darboux2}. The $h$ and $k$ are invariants relative to gauge transformations of the operator and are a complete system of invariants.
They were generalized for dimension three   in the recent work~\cite{athorne2018}. 
It would also be interesting to analyze in a similar way the case of differential operators on the algebra $\Dens(\R^{1|1})$  for  the superline. Compare~\cite{shemya:voronov2016:berezinians}.

\section*{Acknowledgement} This material is based upon work supported by the National Science Foundation under Grant No.1708033.
We are very thankful to Yuri Berest for
the stimulating discussion and providing a reference to the original publication containing the Frobenius theorem on factorization and   the anonymous referee for   remarks that helped to improve the exposition.


\begin{thebibliography}{100}

\bibitem{athorne2018}
Ch. Athorne.
\newblock {Laplace maps and constraints for a class of third-order partial
  differential operators}.
\newblock {\em Journal of Physics A: Mathematical and Theoretical}, 51(8),
  2018.

\bibitem{biggs:lifting}
A.~Biggs.
\newblock {The existence of a canonical lifting of even {P}oisson structures to
  the algebra of densities}.
\newblock {\em Lett. Math. Phys.}, 104(12):1523--1533, 2014.

\bibitem{biggs_hovik2013}
A.~Biggs and H.~Khudaverdyan.
\newblock {Operator pencil passing through a given operator}.
\newblock {\em Journal of Mathematical Physics}, 54(12), 2013.

\bibitem{Darboux2}
G.~Darboux.
\newblock {\em {Le\c{c}ons sur la th{\'e}orie g{\'e}n{\'e}rale des surfaces et
  les applications g{\'e}om{\'e}triques du calcul infinit{\'e}simal}},
  volume~2.
\newblock Gauthier-Villars, 1889.

\bibitem{ovsienko_duval97}
C.~Duval and V.~Yu. Ovsienko.
\newblock {Space of second-order linear differential operators as a module over
  the {L}ie algebra of vector fields}.
\newblock {\em Adv. Math.}, 132(2):316--333, 1997.

\bibitem{frobenius}
G.~Frobenius.
\newblock {Ueber die Determinante mehrerer Functionen einer Variabeln}.
\newblock {\em J. Reine Angew. Math.}, 77:245--257, 1874.

\bibitem{ovsienko_2007_circle}
H.~Gargoubi, N.~Mellouli, and V.~Ovsienko.
\newblock {Differential operators on supercircle: conformally equivariant
  quantization and symbol calculus}.
\newblock {\em Lett. Math. Phys.}, 79(1):51--65, 2007.

\bibitem{jgeorge1}
J.~George.
\newblock {Projective connections and the algebra of densities}.
\newblock In {\em {Geometric methods in physics}}, volume 1079 of {\em {AIP
  Conf. Proc.}}, pages 142--148. Amer. Inst. Phys., Melville, NY, 2008.

\bibitem{tv:darbouxsuperline}
S.~Hill, E.~Shemyakova, and Th. Voronov.
\newblock {Darboux transformations for differential operators on the
  superline}.
\newblock {\em Russian Math. Surveys}, 70(6):1173, 2015.
\newblock \texttt{arXiv:1505.05194 [math.MP]}.

\bibitem{hitchin:hsw}
N.~J. Hitchin, G.~B. Segal, and R.~S. Ward.
\newblock {\em {Integrable systems}}, volume~4 of {\em {Oxford Graduate Texts
  in Mathematics}}.
\newblock The Clarendon Press, Oxford University Press, New York, 2013.
\newblock Twistors, loop groups, and Riemann surfaces, Lectures from the
  Instructional Conference held at the University of Oxford, Oxford, September
  1997, Paperback reprint [of MR1723384].

\bibitem{shemya:hobby2016:iterated}
D.~Hobby and E.~Shemyakova.
\newblock {Classification of multidimensional Darboux transformations: first
  order and continued type}.
\newblock {\em SIGMA (Symmetry, Integrability and Geometry: Methods and
  Applications)}, 13(10):20 pages, 2017.
\newblock \texttt{arXiv:1605.04362 [math.DG]}.



\bibitem{khudaverdian1991geometry}
O.~Khudaverdian.
\newblock {Geometry of superspace with even and odd brackets}.
\newblock {\em J. Math. Phys.}, 32(7):1934--1937, 1991.
\newblock Preprint of the Geneva University, UGVA-DPT 1989/05-613.


\bibitem{khudaverdian1999delta}
H.~Khudaverdian.
\newblock {Delta-Operator on Semidensities and Integral Invariants in the
  Batalin-Vilkovisky Geometry}.
\newblock {\em Preprint 1999/135, Max-Planck-Institut fuer Mathematik Bonn},
  1999.

\bibitem{khudaverdian2004semidensities}
H.~M. Khudaverdian.
\newblock {Semidensities on odd symplectic supermanifolds}.
\newblock {\em Comm. Math. Phys.}, 247(2):353--390, 2004.

\bibitem{hov:mattielap}
H.~Khudaverdian and M.~Peddie.
\newblock {Odd {L}aplacians: geometrical meaning of potential and modular
  class}.
\newblock {\em Lett. Math. Phys.}, 107(7):1195--1214, 2017.

\bibitem{tv:laplace2}
H.~Khudaverdian and Th. Voronov.
\newblock {On odd {L}aplace operators. {II}}.
\newblock In {\em {Geometry, topology, and mathematical physics}}, volume 212
  of {\em {Amer. Math. Soc. Transl. Ser. 2}}, pages 179--205. Amer. Math. Soc.,
  Providence, RI, 2004.

\bibitem{tv:hovgroupoids}
H.~Khudaverdian and Th. Voronov.
\newblock {Geometry of differential operators of second order, the algebra of
  densities, and groupoids}.
\newblock {\em J. Geom. Phys.}, 64:31--53, 2013.

\bibitem{tv:algebradens}
H.~Khudaverdian and Th. Voronov.
\newblock {Geometric constructions on the algebra of densities}.
\newblock In {\em {Topology, geometry, integrable systems, and mathematical
  physics}}, volume 234 of {\em {Amer. Math. Soc. Transl. Ser. 2}}, pages
  241--263. Amer. Math. Soc., Providence, RI, 2014.

\bibitem{shemya:voronov2016:berezinians}
S.~Li, E.~Shemyakova, and Th. Voronov.
\newblock {Differential operators on the superline, Berezinians, and Darboux
  transformations}.
\newblock {\em Lett. Math. Phys.}, 107(9):1689--1714, 2017.

\bibitem{ovsienko-tabachnikov:book2005}
V.~Ovsienko and S.~Tabachnikov.
\newblock {\em Projective differential geometry old and new. From the
  Schwarzian derivative to the cohomology of diffeomorphism groups}, volume 165
  of {\em Cambridge Tracts in Mathematics}.
\newblock Cambridge University Press, Cambridge, 2005.

\bibitem{schwarz1993geometry}
A.~Schwarz.
\newblock {Geometry of Batalin-Vilkovisky quantization}.
\newblock {\em Comm. Math. Phys.}, 155(2):249--260, 1993.

\bibitem{2013:invertible:darboux}
E.~Shemyakova.
\newblock {Invertible {D}arboux transformations}.
\newblock {\em SIGMA (Symmetry, Integrability and Geometry: Methods and
  Applications)}, 9:Paper 002, 10, 2013.

\bibitem{shem:darboux2}
E.~Shemyakova.
\newblock {Proof of the completeness of {D}arboux {W}ronskian formulae for
  order two}.
\newblock {\em Canad. J. Math.}, 65(3):655--674, 2013.

\bibitem{shemyakova2013_DT_fact}
E.~Shemyakova.
\newblock {Factorization of Darboux transformations of arbitrary order for 2D
  Schr{\"o}dinger operator}.
\newblock {\em submitted}, 2017.
\newblock \texttt{arXiv:1304.7063 [math.MP].}

\bibitem{thomas:projtheory}
T.~{Thomas}.
\newblock {A projective theory of affinely connected manifolds}.
\newblock {\em Math. Z.}, 25(1):723--733, 1926.

\end{thebibliography}


\end{document}